\documentclass[onecolumn,11pt,aps,nofootinbib,superscriptaddress,tightenlines,showkeys]{revtex4}

\usepackage[T1]{fontenc}
\usepackage[utf8]{inputenc}

\usepackage{color}

\usepackage[cmex10]{amsmath}
\usepackage{graphics}
\usepackage{dsfont}
\usepackage{amssymb}
\usepackage{graphicx}
\usepackage{mathrsfs}
\usepackage{units}
\usepackage{pgfplots}
\usepackage{enumitem}
\usepackage[colorlinks=true,linkcolor=black,citecolor=black,plainpages=false,pdfpagelabels]{hyperref}
\hypersetup{pdftitle={Eigenvalue estimates for the resolvent of a non-normal matrix}}

\usepackage{amsthm}
\newtheorem{theorem}{Theorem}[section]
\newtheorem{lemma}[theorem]{Lemma}
\newtheorem{proposition}[theorem]{Proposition}
\newtheorem{corollary}[theorem]{Corollary}
\newtheorem{definition}{Definition}[section]

\newenvironment{remark}[1][Remark]{\begin{trivlist}
\item[\hskip \labelsep {\bfseries #1}]}{\end{trivlist}}

\newcommand{\h}{\ensuremath{\mathcal{H}}}

\newcommand{\R}{\mathbb{R}}

\newcommand{\id}{\ensuremath{\mathds{1}}}

\begin{document}
\title{Constrained Quantum Tomography of Semi-Algebraic Sets with Applications to Low-Rank Matrix Recovery}
\author{Michael Kech}
\email{kech@ma.tum.de}
\affiliation{Department of Mathematics, Technische Universit\"{a}t M\"{u}nchen, 85748 Garching, Germany}
\author{Michael M. Wolf}
\email{wolf@ma.tum.de}
\affiliation{Department of Mathematics, Technische Universit\"{a}t M\"{u}nchen, 85748 Garching, Germany}

\date{\today}

\begin{abstract}

We analyze quantum state tomography in scenarios where measurements and states are both constrained. States are assumed to live in a semi-algebraic subset of state space and measurements are supposed to be rank-one POVMs, possibly with additional constraints. Specifically, we consider sets of von Neumann measurements and sets of local observables. We provide upper bounds on the minimal number of measurement settings or outcomes that are required for discriminating all states within the given set. The bounds exploit tools from real algebraic geometry and lead to generic results that do not only show the existence of good measurements but guarantee that almost all measurements with the same dimension characteristic perform equally well.

In particular, we show that on an $n$-dimensional Hilbert space any two states of a semi-algebraic subset can be discriminated by $k$ generic von Neumann measurements if $k(n-1)$ is larger than twice the dimension of the subset. In case the subset is given by states of rank at most $r$, we show that $k$ generic von Neumann measurements suffice to discriminate any two states provided that $k(n-1)>4r(n-r)-2$. {We obtain corresponding results for low-rank matrix recovery of hermitian matrices in the scenario where the linear measurement mapping is induced by tight frames.}

\end{abstract}

\maketitle

\tableofcontents

\section{Introduction}
{
Let us note in the beginning that the reader mainly interested in low-rank matrix recovery can find our corresponding results in Section IV. There we find linear measurement mappings induced by tight frames that can discriminate any two matrices of rank at most $r$.
}

Quantum state tomography, which aims at identifying quantum states from the outcomes of an experiment, is a central task in quantum information science. Full state tomography is often challenging and sometimes infeasible. However, if there is some prior information about the state under investigation, this can considerably simplify the problem: the number of measurement settings necessary to uniquely identify a given state can significantly decrease if the state is not arbitrary but is known to lie on a confined subset of state space.

Using topological properties of the measurement map and the constrained set, lower bounds on the  minimal number of measurement settings necessary to discriminate any two pure states were obtained in \cite{heinosaari2013quantum}. Relating these topological features of the measurement map to stability properties, it was shown in \cite{kech} that under the premise of stability the approach of \cite{heinosaari2013quantum} can be generally applied. Using this result, lower bounds on the necessary number of measurement settings for several other subsets were obtained in \cite{kech}.

The present article deals with the issue of finding upper bounds: given a subset of state space, find a measurement scheme that can discriminate any two states of this subset with as few measurement settings as possible. This appears to be a rather hard problem in general. Already in the case of pure state quantum tomography it has received significant attention in  topology \cite{milgram1967immersing,mayer}, quantum information science \cite{weigert1992pauli,amiet1999reconstructing,amiet1998reconstructing2,finkelstein2004pure,flammia2005minimal,heinosaari2013quantum,gross2010quantum,mondragon2013determination,CarmeliTeikoJussi1,CarmeliTeikoJussi2,flammia2005minimal,QCS2} and sampling theory \cite{balan2006signal,conca2014algebraic,bodmann2013stable,GrossPhaseLift}. 

In addition to constraining the set of states, we also restrict the set of measurements in order to capture the fact that arbitrary measurements may not be feasible in an experiment. The imposed constraints could for example be the restriction to von Neumann measurements or to local measurements when dealing with a multipartite system. The case of pure state tomography with von Neumann measurements was addressed in \cite{mondragon2013determination,jaming2010uniqueness,carmeli2015many}. In \cite{mondragon2013determination,jaming2010uniqueness} it was shown that any two pure states can be discriminated by merely $4$ von Neumann measurements. This is known to be sharp for pure states of an $n$-dimensional Hilbert space if $n>4$ and \cite{carmeli2015many} has a special focus on the cases $n\leq 4$. The more general setting of low-rank matrix recovery with restricted measurements was considered in \cite{kueng2014low}. However, their focus is to determine the asymptotic behaviour, and this allows us to improve on some of their results.

We propose a method that can deal with these problems rather generally and we then apply it to different scenarios.

In this article we neither consider the statistical aspects of quantum tomography nor the algorithmic problem of reconstructing the state from the measurement data. 

\textit{Outline.} In Section II we fix notation, introduce measurement schemes that are relevant in the following and give some preliminary results about hermitian matrices of bounded rank. Furthermore, we illustrate the connection between phase retrieval and quantum tomography.

In Section III,  we propose a method to find sets of measurements that can discriminate any two states of a given subset of the state space, generalizing the approach taken in \cite{balan2006signal} to find frames for the phase retrieval problem. The method can be applied to all semi-algebraic subsets and it can naturally deal with constrained measurement like e.g. von Neumann measurements. Rather than giving explicit constructions, the method asserts that almost all sets of measurement that fulfil certain constraints allow for a unique identification. 

{
In Section IV, we apply this procedure  to low-rank matrix recovery, showing that a generic frame with $m>4r(n-r)$ frame vectors can discriminate any two hermitian matrices of rank at most $r$. This generalizes \cite{conca2014algebraic} where the case $r=1$ was considered. In addition we shown that the statement also holds when restricting to tight frames.}

In Section V, we prove that under a further condition the sets of measurements obtained by the method introduced in Section III fulfil the stability property introduced in \cite{kech}. In the scenarios where the method is feasible this condition is satisfied and therefore the stability property holds rather generally.

In Section VI, we present the main result of this paper. Loosely speaking, it asserts that one can perform tomography on all semi-algebraic subsets of the state space by measuring sets of positive operator valued measures (POVMs) that consist exclusively of rank one operators, in particular von Neumann measurements. From this result we straightforwardly obtain Whitney type embedding results for these measurement schemes. Furthermore, we consider the problem of discriminating states of bounded rank: In \cite{heinosaari2013quantum,kech} lower bounds on the number of measurement outcomes necessary to uniquely identify quantum states with bounded rank were established and these lower bounds turned out to be close to the upper bounds obtained in \cite{heinosaari2013quantum} where it was shown that $4r(n-r)$ measurement outcomes suffice in order to identify states of an $n$-dimensional system with rank at most $r$.  However, the measurement that does realize this upper bound has a rather complicated structure. We prove that the same upper bounds as in \cite{heinosaari2013quantum} can be realized when measuring a POVM which exclusively consist of rank one operators and we prove similar results for measuring sets of von Neumann measurements. Note that our results come with less measurement outcomes than the compressed sensing approach of \cite{gross2010quantum}, however we do not provide a tractable reconstruction procedure.

Section VII deals with the problem of reconstructing states of multipartite systems from the expectation values of local observables. Just like in Section V, we first give a theorem stating that one can do tomography on all semi-algebraic subsets of the state-space by performing measurements of this type. Then we obtain Whitney type embedding results and also for the problem of identifying states of bounded rank we obtain corresponding results.

In Section VIII, proofs of technical results are given.

Most of our results assert that almost all measurements have a certain property. In the Appendix we discuss the measure with respect to which this is true.

\section{Preliminaries}

Throughout $\mathcal{H}$ denotes a finite-dimensional complex Hilbert space. $H(\mathcal{H})$ denotes the real vector space of hermitian operators\footnote{{ We denote the adjoint of a linear operator $B:\mathcal{H}_1\to\mathcal{H}_2$ by $B^\dagger$.}} on $\mathcal{H}$ and $\mathcal{S}(\mathcal{H})$ denotes the set of quantum states on $\mathcal{H}$, i.e. $\mathcal{S}(\mathcal{H})=\{\varrho\in H(\mathcal{H}):\varrho\geq 0, \text{tr}(\varrho)=1\}$. We regard $H(\h)$ as an inner product space, equipping it with the Hilbert-Schmidt inner product. The Hilbert Schmidt norm is denoted by $\|\cdot\|_2$. By $SH(\h):=\{X\in H(\h):\ \|X\|_2^2=\text{tr}(X^2)=1\}$ we denote the unit sphere in $H(\h)$. Furthermore, for a subset $A\subseteq H(\h)$, $\Delta(A)$ denotes the set of differences of operators in $A$, i.e. $\Delta(A)=\{X-Y:X,Y\in A\}$.  $M(m,n,\mathbb{C})\ (M(m,n,\R))$ denotes the set of complex (real) $m\times n$ matrices and we write $M(n,\mathbb{C})\ (M(n,\R))$ as shorthand for $M(n,n,\mathbb{C})\ (M(n,n,\R))$. 

{
In the following, measurements are modelled by linear mappings form the set of hermitian operators (respectively hermitian matrices) to $\mathbb{R}^m$, where $m$ is the number of measurement outcomes.
\begin{definition}[Measurement map]
A linear mapping $h:H(\mathcal{H})\to\mathbb{R}^m$ is called a measurement map. The number of outcomes of $h$ is $m$.
\end{definition}
}

\subsection*{(Constrained) Measurements in Quantum Mechanics}
{ In this section we focus on the specific  measurement maps that typically arise in quantum mechanics.}
In quantum mechanics POVMs are used to describe general measurements \cite{holevo2011probabilistic, busch1995operational}. For the purpose of this article a POVM on $\h$ is a tuple $P=(Q_{1},\hdots,Q_{m})$ of positive semidefinite operators on $\h$ such that
\begin{align*}
\sum_{i=1}^{m}Q_{i}=\id_{\mathcal{H}}.
\end{align*}
An element of $P$ is called an effect operator. We define the dimension of $P$ by $\dim P:=|P|$. 

A whole measurement scheme might consist of measuring more than one POVM.
\begin{definition}
A measurement scheme on $\h$ is a tuple $M=(P_1,\hdots,P_k)$ of POVMs on $\h$. We define the dimension of $M$ by $\dim M:=\dim P_1+\hdots+\dim P_k$. 
\end{definition} 
A POVM $P$ can be identified with the measurement scheme that just contains $P$. In the following we sometimes make use of this identification and regard POVMs as measurement schemes.

A POVM $P=(Q_{1},\hdots,Q_{m})$ induces a measurement map
\begin{align*}
h_{P}:H(\mathcal{H})&\to \mathbb{R}^{m} \\
  X&\mapsto\big( \text{tr}(Q_{1}X),\hdots,\text{tr}(Q_{m}X) \big).
\end{align*}
Similarly a measurement scheme $M=(P_1,\hdots,P_k)$ induces a measurement map
\begin{align*}
h_{M}:H(\mathcal{H})&\to \mathbb{R}^{|P_1|+\hdots+|P_k|} \\
  X&\mapsto\big( h_{P_1}(X),\hdots,h_{P_k}(X) \big).
\end{align*}

\begin{definition}\label{defcomplete}
A measurement scheme $M$ is called $\mathcal{R}$-complete for a subset $\mathcal{R}\subseteq \mathcal{S}(\mathcal{H})$ if $h_{M}|_{\mathcal{R}}$ is injective.
\end{definition}

Our main results are statements about rank one POVMs and von Neumann measurements, so let us define these terms: A POVM $P$ is called rank one POVM if all effect operators are of rank one. We denote the set of $m$-dimensional rank one POVMs on $\h$ by $\mathcal{M}_1^m(\h)$. In the following we implicitly assume that $m\geq \dim\h$ because otherwise $\mathcal{M}_1^m(\h)$ would be empty.

Later on we often use the following correspondence between linear isometries and $\mathcal{M}_1^m(\mathbb{C}^{n})$: The equations 
\begin{align*}
M^\dagger M=\id_n,\ M\in M(m,n,\mathbb{C}),
\end{align*}
can be considered as real algebraic equations under the identification $M(m,n,\mathbb{C})\simeq \R^{2nm}$. The solution set $U(m,n)$ is the set of linear isometries $U:\mathbb{C}^n\to\mathbb{C}^m$. Note that $U(m,n)$ is non-empty if and only if $m\geq n$ and that for $n=m$ it is the set of unitaries. We write $U(n)$ as shorthand for $U(n,n)$.

Let $\{e_i\}_{i\in\{1,\hdots,m\}}$ be the standard basis of $\mathbb{C}^m$. Then, the sought correspondence is given by the map
\begin{align}\label{phi}
\begin{split}
\phi:U(m,n)&\to\mathcal{M}_{1}^m(\mathbb{C}^{n})\\
 U&\mapsto (U^\dagger e_1e_1^\dagger U,\hdots,U^\dagger e_1e_1^\dagger U).
 \end{split}
\end{align}

If the effect operators of a POVM are projections on mutually orthogonal subspaces, the POVM is called von Neumann measurement. In this article, we just deal with rank one von Neumann measurements and therefore, in the following, the term von Neumann measurement always refers to rank one von Neumann measurements. Note, that the set of rank one von Neumann measurements is precisely the set of $(\dim{\h})$-dimensional rank one POVMs.

The measurement scheme consisting of $k$ $m$-dimensional rank one POVMs on $\h$ is denoted by $\mathcal{M}^m_{1,k}(\h)$, i.e.
\begin{align*}
\mathcal{M}^m_{1,k}(\h)=\{(P^1,\hdots,P^k): P^i\in\mathcal{M}_{1}^{m}(\h)\}.
\end{align*}
For $m=\dim\h$ this is the set of $k$ rank one von Neumann measurements which, we denote by
$\mathcal{M}^k_{\text{vN}}(\h)$.

\subsection*{Hermitian Matrices of Bounded Rank}
In this section we prove a lemma about hermitian operators with bounded rank, which is frequently used in the following. Denote by $\mathcal{P}_r(\h)$ the set of hermitian operators on $\h$ with rank at most $r$, i.e. $\mathcal{P}_r(\h):=\{X\in H(\h): \text{rank}(X)\leq r\}$. We write $\mathcal{P}_{r}^{n}$ as shorthand for $\mathcal{P}_{r}(\mathbb{C}^{n})$.
 
\begin{lemma}\label{lemrank}
$\mathcal{P}_{r}^{n}$ is a real algebraic set of dimension $r(2n-r)$.
\end{lemma}
\begin{proof}
First note that $\mathcal{P}_{r}^{n}$ is a real algebraic set: It is given by the set of points $X\in M(n,\mathbb{C})$ for which all $(r+1)\times (r+1)$-minors vanish and that satisfy $X=X^{\dagger}$. These conditions turn into a set of real algebraic equations under the canonical identification $M(n,\mathbb{C})\simeq\mathbb{R}^{2n^2}$.

To determine the dimension of $\mathcal{P}_{r}^{n}$ consider the semi-algebraic set $V_{r}^{n}=\{(P_{1},\hdots,P_{r}):P_i\in\mathcal{P}_{1}^{n},\ \text{tr}(P_{i}P_{j})=\delta_{ij},\ P_i\geq 0\}$ \footnote{A hermitian matrix is positive semidefinite if and only if all of its principal minors are greater than of equal to zero. Thus, the equations $P_i\geq 0$ can be regarded as algebraic inequalities.}. The dimension of $V_{r}^{n}$ is given by $r(2n-r)-r$. To see this, consider the smooth and transitive action of $U(n)$ on the complex matrices $M(n,\mathbb{C})$ given by $(U,M)\to (U,UMU^\dagger)$ and let $V_D$ be the orbit of the diagonal matrix $D:=\text{diag}(r,r-1,\hdots,1,0,\hdots)$ under this action. Noting that the stabilizer subgroup of $D$ is $U(n-r)\times U(1)^r$ we obtain $V_D\simeq U(n)/(U(n-r)\times U(1)^r)$ by Theorem 3.62 of \cite{warner1971foundations}. But the semi-algebraic map $\psi:V_r^n\to V_D,\ (P_{1},\hdots,P_{r})\mapsto \sum_{j=1}^rjP_j$ is clearly bijective. Hence we find $\dim V_r^n=\dim(U(n)/(U(n-r)\times U(1)^r))=n^2-(n-r)^2-r=r(2n-r)-r$ by Theorem 2.8.8 and Proposition 2.8.14 of \cite{bochnak1998real}.


The semi-algebraic map
\begin{align}\label{1}
\begin{split}
\eta:\R^r\times V_{r}^{n}&\to\mathcal{P}_{r}^{n}\\
(\lambda_1,\hdots,\lambda_r,P_{1},\hdots,P_{r})&\mapsto \sum_{i=1}^{r} \lambda_i P_i.
\end{split}
\end{align}
is clearly surjective. By Theorem 2.8.8 of \cite{bochnak1998real}, we hence conclude that $\dim \mathcal{P}_{r}^{n}\leq\dim V_{r}^{n}+r=r(2n-r)$ and furthermore that indeed $\dim \mathcal{P}_{r}^{n}=r(2n-r)$ by noting that $\phi$ is injective if we require $\lambda_1>\hdots>\lambda_r>0$.
\end{proof}
\begin{corollary}\label{corrank}
The set $\mathcal{D}_1:=\{X\in\mathcal{P}_{r}^{n}:\ \text{tr}(X^2)=2\}$ is a real algebraic set of dimension $r(2n-r)-1$ and the set $\mathcal{D}_2:=\{X\in\mathcal{P}_{r}^{n}:\ \text{tr}(X^2)=2,\ \text{tr}(X)=0\}$ is a real algebraic set of dimension $r(2n-r)-2$.
\end{corollary}
\begin{proof}
From the proof of  Lemma \ref{lemrank} it is immediate that both $\mathcal{D}_1$  and $\mathcal{D}_2$  are real algebraic sets. To determine the dimension of $\mathcal{D}_1$, one can go along the lines of the proof of Lemma \ref{lemrank} and simply replace $\R^n$ by the unit sphere $S^{n-1}$ in the definition of the mapping $\eta$. Similarly, to determine the dimension $\mathcal{D}_2$, one can go along the lines of the proof of Lemma \ref{lemrank} and this time replace $\R^n$ by $\{x\in S^{n-1}:\sum_{i=1}^nx_i=0\}$ in the definition of the mapping $\eta$.
\end{proof}
\subsection*{Frames and Rank One POVMs}

Finally, we discuss the connection between pure state tomography and the phase retrieval problem in sampling theory. A finite set $F=\{v_1,\hdots,v_m\}$ of vectors in $\mathbb{C}^n$ is called a frame if there exist constants $a,b>0$ such that
\begin{align}\label{eqframe}
a\|x\|_2^2\leq\sum_{i=1}^{m}|\langle x,v_i\rangle|^2\leq b\|x\|_2^2\ \text{for all}\ x\in\mathbb{C}^n.
\end{align}
A frame $F=\{v_1,\hdots,v_m\}$ induces a measurement map
\begin{align}\label{measframe}
\begin{split}
M_F:\mathbb{C}^n\delimiter"502F30E\mathopen{{\sim}}&\to \R^m\\
[x]&\mapsto (|\langle v_1,x\rangle|^2,\hdots,|\langle v_m,x\rangle|^2)
\end{split}
\end{align}
where $x\sim y$ iff there is a $\lambda\in\R$ such that $x=e^{i\lambda}y$ \footnote{Note that $M_F$ is also well-defined for $F=\{v_1,\hdots,v_m\}$ with $v_i\in\mathbb{C}^n$, i.e. if we do not require $F$ to be a frame.}.  Since the task in phase retrieval is to reconstruct signals modulo phase from intensity measurements, one considers frames $F$ such that $M_F$ is injective.

Each frame $F=\{v_1,\hdots,v_m\}$ induces a map
\begin{align}\label{map}
\begin{split}
h_{F}:H(\mathbb{C}^n)&\to\R^m\\
X&\to(\text{tr}(Xv_1 v_1^\dagger),\hdots,\text{tr}(X v_m v_m^\dagger)).
\end{split}
\end{align}
Noting that $h_{F}(x x^\dagger)=M_F(x)$, we conclude that $h_{P_F}|_{\mathcal{P}^n_1}$ is injective if and only if $M_F$ is injective. 

A corollary of one of our main results is a statement about tight frames, so let us define this term. A frame $F$ is called tight frame if $a=b$ in inequality \eqref{eqframe}. If in addition $a=b=1$, $F$ is called tight frame.

The following proposition shows the well-known fact that tight frames correspond to rank one POVMs.
\begin{proposition}\label{proppars}
Let $F$ be a tight frame. Then the associated set of rank one operators $P_F$ is a POVM.
\end{proposition}
\begin{proof}
Let $F=\{v_1,\hdots,v_m\}$. Since $F$ is a tight frame, we obtain the following equality from inequality \eqref{eqframe}:
\begin{align*}
\sum_{i=1}^m|\langle v_i,x\rangle|^2=\|x\|_2^2.
\end{align*}
This can be rewritten as
\begin{align*}
\sum_{i=1}^m|\langle v_i,x\rangle|^2=\text{tr}(x x^\dagger\sum_{i=1}^mv_i v_i^\dagger)=\|x\|^2.
\end{align*}
But since this holds for all $x\in\mathbb{C}^n$ we conclude that $\sum_{i=1}^mv_i v_i^\dagger=\id_{\mathbb{C}^n}$: Assume $\sum_{i=1}^m v_i v_i^\dagger\neq \id_{\mathbb{C}^n}$. Since $\sum_{i=1}^mv_i v_i^\dagger$ is hermitian there has to be an eigenvector $w$ of $\sum_{i=1}^mv_i v_i^\dagger$ with eigenvalue $\lambda\neq 1$. But then $ w^\dagger\sum_{i=1}^m v_i v_i^\dagger w=\lambda \|w\|_2^2\neq \|w\|_2^2$, a contradiction.
\end{proof}
\begin{remark}
Note that the correspondence is given by the map $\phi$ defined in equation \eqref{phi} where the frame vectors are given by the rows of the isometry.
\end{remark}
Let $P$ be a POVM. In pure state tomography, not $h_{P}|_{\mathcal{P}^n_1}$ is required to be injective, but $h_{P}|_{\mathcal{S}_1^n}$ where $\mathcal{S}_1^n:=\{\varrho\in\mathcal{S}(\mathbb{C}^n):\varrho^2=\varrho\}$ is the set of pure states. However, by the definition of a POVM, $\id_n\in P$ and this implies that if $h_{P}|_{\mathcal{S}_1^n}$
is injective, also $h_{P}|_{\mathcal{P}_1^n}$ is injective. From this point of view, pure state quantum tomography with rank one POVMs is equivalent to phase retrieval with tight frames.

\section{The Basic Idea}\label{method}

Let us begin by explaining the basic idea of the method we utilize to find one-to-one measurement schemes which originates from the approach taken in \cite{balan2006signal} to find frames for the phase retrieval problem.

The method essentially relies on the following observation: A measurement scheme $P:=((Q^1_1,\hdots,Q^1_m),\hdots,(Q^k_1,\hdots,Q^k_m))$ is $\mathcal{R}$-complete with respect to a subset $\mathcal{R}\subseteq \mathcal{S}(\h)$ if and only if the equations 
\begin{align}\label{constf}
\text{tr}(Q^j_i X)=0,\ \ \,\ i\in\{1,\hdots,m-1\},j\in\{1,\hdots,k\}
\end{align}
have no solution for $X\in\Delta(\mathcal{R})-\{0\}$.

For a given subset $\mathcal{R}\subseteq\mathcal{S}(\h)$, we want to characterize non-injective measurement schemes via the equations \eqref{constf} and use the dimension theory of semi-algebraic sets to show that these have measure zero. Therefore, we consider measurement schemes that are constrained by real algebraic equalities or inequalities. In the following, the set of measurement schemes is a semi-algebraic set $\mathcal{M}$ such that for all $M\in\mathcal{M}$ we have $\dim P=m, \forall P\in M$ and $|M|=k$ where $m,k\in\mathbb{N}$ are some fixed numbers. For example, if $k=1$, this could be the restriction to the set of $m$-dimensional rank one POVMs $\mathcal{M}_1^m(\h)$. Furthermore, in order to ensure that the equations \eqref{constf} in fact become real algebraic equations, we have to replace $\Delta(\mathcal{R})-\{0\}$ by a suitable semi-algebraic set. We do this by constructing a semi-algebraic set $\mathcal{D}\subseteq H(\h)$\footnote{Here we identify $H(\h)$ with $(\dim\h)^2$-dimensional real affine space.} with the following property: If there is a measurement scheme $M$ and an $X\in\Delta(\mathcal{R})-\{0\}$ with
\begin{align}\label{const}
h_M(X)=0
\end{align}
then there exists $X^\prime\in\mathcal{D}$ with
\begin{align}\label{constff}
h_M(X^\prime)=0.
\end{align}
If a semi-algebraic set $\mathcal{D}\subseteq H(\h)$ with $0\notin\mathcal{D}$ has this property, we say that $\mathcal{D}$ represents $\Delta(\mathcal{R})-\{0\}$. 

The solution set of the equations \eqref{constff} characterizes the non-injective measurement schemes: Let $\tilde{\mathcal{M}}$ be the real semi-algebraic set obtained from $\mathcal{M}\times\mathcal{D}$ by imposing the equations \eqref{constff}. By construction of $\mathcal{D}$, the non-injective measurement schemes are contained in the projection of $\tilde{\mathcal{M}}\subseteq\mathcal{M}\times\mathcal{D}$ on the first factor with the canonical projection $\pi_1:\mathcal{M}\times\mathcal{D}\to\mathcal{M}$. But if $\dim \tilde{\mathcal{M}}<\dim\mathcal{M}$, we also have $\dim\pi_1( \tilde{\mathcal{M}})<\dim\mathcal{M}$ \footnote{$\pi_1$ maps semi-algebraic sets to semi-algebraic sets and does not increase the dimension. See Theorem 2.2.1 and Proposition 2.8.6 of \cite{bochnak1998real}.} and thus the non-injective measurement schemes have measure zero in $\mathcal{M}$. Here we used the well-know fact that, for a suitably chosen measure, the measure of a semi-algebraic subset $S$ of a semi-algebraic set $A$ has measure zero in $A$ if $\dim A>\dim S$.  For more details on the measure see Appendix \ref{appendixD}.

This approach is most efficient if the equations \eqref{constff} are transversal to $\mathcal{M}\times\mathcal{D}$. In this case $\dim \tilde{\mathcal{M}}<\dim\mathcal{M}$ is equivalent to $k(m-1)>\dim \mathcal{D}$ and thus the quality of our result is determined by how low-dimensional we can choose the semi-algebraic set $\mathcal{D}$.

\section{Low-Rank Matrix Recovery with Frames}

To illustrate how this procedure works, let us consider the problem of low-rank matrix recovery with frames. We show that any two hermitian matrices of rank at most $r$ can be discriminated from a generic frame with $m\geq 4r(n-r)$ frame vectors. The proof we give is inspired by the proof of Theorem 3.1 in \cite{balan2006signal}. Let $r\in\{1,\dots,[n/2]\}$\footnote{Here $[x]:=$largest integer $i$ such that $i\leq x$.}.
\begin{theorem}[Low-Rank Matrix Recovery with Frames]\label{frame}
Let $m\geq 4r(n-r)$. For almost all frames $F=\{v_1,..,v_{m}\}$ the map $h_F|_{\mathcal{P}_r^n}$ (see Equation \eqref{map}) is injective.
\end{theorem}
\begin{proof}
Let $F=(v_1,\hdots,v_m),\ v_i\in\mathbb{C}^n,$ and consider the equations
\begin{align}\label{eqap}
v_i^\dagger X v_i=0,\ i\in\{1,\hdots,m\},
\end{align}
in $v_i\in\mathbb{C}^n$, $X\in\Delta(\mathcal{P}_r^n)-\{0\}$. As explained above, these equations determine the subset $N$ of $F\in\mathbb{C}^{nm}\simeq \mathbb{R}^{2nm}$ for which $h_{F}|_{\mathcal{P}_r^n}$ fails to be injective. 

Note that $\Delta(\mathcal{P}_r^n)-\{0\}=\mathcal{P}_{2r}^n-\{0\}$. Consider the algebraic set $\mathcal{D}:=\{X\in\mathcal{P}_{2r}^n:\text{tr}(X^2)=1\}$ and note that  we have $\dim{\mathcal{D}}=4r(n-r)-1$ by Corollary \ref{corrank}. Furthermore, $\mathcal{D}$ represents $\Delta(\mathcal{P}_r^n)-\{0\}$: 
Clearly $0\notin\mathcal{D}$. Next, consider a measurement scheme $M$ and $X\in\mathcal{P}_{2r}^n-\{0\}$ such that $h_M(X)=0$. But then there is $X^\prime:=\frac{X}{\|X\|_2}\in\mathcal{D}$ such that $h_M(X^\prime)=\frac{1}{\|X\|_2}h_M(X)=0$.

Under the identification $\mathbb{C}^{nm}\simeq \mathbb{R}^{2nm}$ the equations \eqref{eqap} are $m$ equations on the real algebraic set $\mathbb{C}^{nm}\times\mathcal{D}$ and next we prove that imposing these equations decreases the dimension of $\mathbb{C}^{nm}\times\mathcal{D}$ by at least $m$: Note that it suffices to prove that imposing the equation \eqref{eqap} on $\mathbb{C}^{nm}$, for fixed $X\in\mathcal{D}$, decreases the dimension by at least $m$. But for fixed $X\in\mathcal{D}$, the $i$-th equation of \eqref{eqap} just involves the variables of the $i$-th factor in $(\mathbb{C}^{n})^m$. Thus it suffices to prove that for given $X\in\mathcal{D}$ imposing the equation
\begin{align}\label{eqnframe}
p(v):= v^\dagger X v=0,\ v\in\mathbb{C}^n,
\end{align}
on $\mathbb{C}^n\simeq\R^{2n}$ decreases the dimension by at least one. But for given $X\in\mathcal{D}$ there is $v\in\mathbb{C}^n$ such that $p(v)= v^\dagger X v=\text{tr}(X v v^\dagger)\neq 0$ because $H(\mathbb{C}^n)$ has a basis of rank one operators and $X\neq 0$. Thus, \eqref{eqnframe} is a non-trivial algebraic equation on the irreducible algebraic set $\mathbb{C}^n\simeq\R^{2n}$. But this immediately implies that \eqref{eqnframe} does decrease the dimension \footnote{Every proper algebraic subset of the irreducible algebraic set $\R^{2m}$ has dimension less than $2m$.}.

Let $\mathcal{M}$ be the algebraic subset of $\mathbb{C}^{nm}\times\mathcal{D}$ obtained by imposing the equations \eqref{eqap} and denote by $\pi_1: \mathbb{C}^{nm}\times\mathcal{D}\to\mathbb{C}^{nm}$ the canonical projection on the first factor. For $m>\dim \mathcal{D}=4r(n-r)-1$, we find $\dim \pi_1(\mathcal{M})<\dim \mathbb{C}^{nm}=2nm$ since imposing the equations \eqref{eqap} on $\mathbb{C}^{nm}$ decreases the dimension by at least $m$. Thus, we conclude that  $\pi_1(\mathcal{M})$  has Lebesgue measure zero \footnote{The Lebesgue measure on $\R^n$ is a rescaling of the $n$-dimensional Hausdorff-measure.} in $\mathbb{C}^{nm}$. Hence, the subset of $F\in\mathbb{C}^{nm}$ for which $h_{F}|_{\mathcal{P}_r^n}$ is injective has full Lebesgue measure. Note, that the subset of frames in $\mathbb{C}^{nm}$ has full Lebesgue measure for $m\geq n$. Choosing the measure on the set of frames to be the restriction of the Lebesgue measure, also the subset of frames for which $M_F$ is injective has full measure.
\end{proof}

For $r=1$, this is the phase retrieval problem and in this case Theorem \ref{frame} reproduces the main result of \cite{conca2014algebraic}.
\begin{corollary}
Let $m\ge 4n-4$. For almost all frames $F=\{v_1,..,v_{m}\}$ the map $M_F$ (see Equation \eqref{measframe}) is injective.
\end{corollary}
\begin{proof}
Let $F=\{v_1,\hdots,v_m\},\ v_i\in\mathbb{C}^n,$ and consider the equations
\begin{align*}
|\langle v_i,x\rangle|^2-|\langle v_i,y\rangle|^2= v_i^\dagger(x x^\dagger-y y^\dagger)v_i=0,\ i\in\{1,\hdots,m\},
\end{align*}
in $x,y,v_i\in\mathbb{C}^n$ where $x x^\dagger-y y^\dagger\neq 0$. These equations determine the subset $N$ of $F\in\mathbb{C}^{nm}\simeq \mathbb{R}^{2nm}$ for which $M_F$ fails to be injective. It is easily seen that the equations
\begin{align}\label{opo}
 v_i^\dagger Xv_i=0,\ i\in\{1,\hdots,m\},
\end{align}
where $X\in\Delta(\mathcal{P}_1^n)-\{0\}$, determine the same subset $N$. But the equations \eqref{opo} are precisely the equations \eqref{eqap} for $r=1$. Thus, the proof can be concluded by going along the lines of the proof of Theorem \ref{frame}.
\end{proof}

{ A similar result holds true for tight frames.
\begin{theorem}[Low-Rank Matrix Recovery with Tight Frames]\label{tightframe}
If $k(m-1)\ge 4r(n-r)-1$, then for almost all collections of tight frames $F_1,\hdots,F_k$, with $|F_i|=m$ for all $i\in\{1,\hdots,k\}$, the map $(h_{F_1},\hdots,h_{F_k})|_{\mathcal{P}_r(\mathbb{C}^n)}$ is injective.
\end{theorem}
The proof of this Theorem relies on Lemma \ref{LEMvN} which is our main technical result. Therefore we relegate its proof to Section \ref{proof}.
}

\section{Stability}
The measurement schemes obtained by the method presented in Section \ref{method} typically come with a stability property. Let
\begin{align*}
\mathcal{M}(n_1,\hdots,n_k):=\{M:=(P^1,\hdots,P^k):\ P^i\text{ POVM with }\dim P^i=n_i\}.
\end{align*}
In this section we denote $\mathcal{M}(n_1,\hdots,n_k)$ by $\mathcal{M}$. We equip $\mathcal{M}$ with the topology induced by the metric
\begin{align*}
d(M,M^\prime):=\|h_M-h_{M^\prime}\|=\sup_{X\in H(\mathbb{C}^n)}\frac{\|h_M(X)-h_{M^\prime}(X)\|_2}{\|X\|_2}
\end{align*}
where $M,M^\prime\in\mathcal{M}$.

\begin{definition}\label{defstab}
Let $\mathcal{R}\subseteq\mathcal{S}(\mathbb{C}^n)$ be a subset. An $\mathcal{R}$-complete measurement scheme $M\in \mathcal{M}$ is stably $\mathcal{R}$-complete if there exists a neighbourhood $\mathcal{N}$ of $M$ such that every measurement scheme $M^\prime\in \mathcal{N}$ is $\mathcal{R}$-complete.
\end{definition}

Let $\mathcal{R}\subseteq\mathcal{S}(\mathbb{C}^n)$ be a subset and let $\mathcal{D}\subseteq\mathcal{S}(\mathbb{C}^n)$ be a semi-algebraic set that represents $\Delta(\mathcal{R})-\{0\}$. Consider the semi-algebraic map
\begin{align}\label{psi}
\begin{split}
\psi: \mathcal{D}&\to H(\mathbb{C}^n)\\
    X&\mapsto \frac{X}{\|X\|_2}. 
\end{split}
\end{align}
By Proposition 2.2.7 and Theorem 2.8.8 of \cite{bochnak1998real}, $\tilde{\mathcal{D}}:=\psi(\mathcal{D})$ is semi-algebraic with $\dim\tilde{\mathcal{D}}\leq\dim\mathcal{D}$. Furthermore $\tilde{\mathcal{D}}$ clearly represents  $\Delta(\mathcal{R})-\{0\}$.

\begin{lemma}\label{stab}
If $\tilde{\mathcal{D}}$ is closed, every $\mathcal{R}$-complete measurement scheme $M\in\mathcal{M}$ is stably $\mathcal{R}$-complete.
\end{lemma}
\begin{proof}
Note that $\tilde{\mathcal{D}}\subseteq SH(\mathbb{C}^n)$. $SH(\mathbb{C}^n)$ is compact and thus $\tilde{\mathcal{D}}$ is compact being a closed subset of a compact set. By the continuity of the induced map $h_M$ and compactness of $\tilde{\mathcal{D}}$, $\kappa:=\min_{X\in\tilde{\mathcal{D}}}\|h_M(X)\|_2$ exists and $\kappa>0$ since $M$ is $\mathcal{R}$-complete. Now let $B(M,\kappa/2):=\{M^\prime\in \mathcal{M}:\sup_{X\in SH(\mathbb{C}^n)}\|h_M(X)-h_M^\prime(X)\|_2<\kappa/2\}$ and note that $B(M,\kappa/2)$ is open. But then
\begin{align*}
\min_{X\in\tilde{\mathcal{D}}}\|h_{M^\prime}(X)\|&\geq \min_{X\in\tilde{\mathcal{D}}}\|h_M(X)\||-\min_{X\in\tilde{\mathcal{D}}}\|h_{M^\prime}(X)-h_M(X)\|\\
&\geq \min_{X\in\tilde{\mathcal{D}}}\|h_M(X)\||-\max_{X\in\tilde{\mathcal{D}}}\|h_{M^\prime}(X)-h_M(X)\|\\
&\geq \kappa-\max_{X\in SH(\mathbb{C}^n)}\|h_{M^\prime}(X)-h_M(X)\|\\
&\geq \kappa-\kappa/2=\kappa/2.
\end{align*}
Thus all measurement schemes $M^\prime\in B(M,\kappa/2)$ are $\mathcal{R}$-complete. 
\end{proof}
\begin{remark}
Note that $\tilde{\mathcal{D}}$ need not be closed for this lemma to apply: In the situations presented in the following the conclusions solely depend on the dimension of $\tilde{\mathcal{D}}$. By Proposition 2.8.2 of \cite{bochnak1998real} the dimension of $\tilde{\mathcal{D}}$ coincides with the dimension of its closure $\overline{\tilde{\mathcal{D}}}$ in the norm topology on $H(\mathbb{C}^n)$. Furthermore, by Proposition 2.2.2 of \cite{bochnak1998real}, the closure of a semi-algebraic set is semi-algebraic. Thus $\overline{\tilde{\mathcal{D}}}$ represents $\Delta(\mathcal{R})-\{0\}$ and $\dim\overline{\tilde{\mathcal{D}}}\le\dim\mathcal{D}$.
\end{remark}

\section{ Quantum Tomography with von Neumann Measurements}

\subsection*{Universality of Rank One POVMs}
The following lemma is the main technical result of this article. It asserts that the equations \eqref{constff} are independent when restricting to rank one POVMs. More precisely let $\h=\mathbb{C}^n$ and denote by $\{e_i\}_{i\in\{1,\hdots,n\}}$ the standard basis of $\mathbb{C}^n$. 

For a fixed non-zero $X\in H(\mathbb{C}^n)$, consider the equations
\begin{align}\label{eqnblaaaa}
\begin{split}
p^j_i(M_1,\dots,M_k)&:=\text{tr}(M_i^\dagger e_j e_j^\dagger M_i X)=e_j^\dagger M_iXM_i^\dagger  e_j=0,\\
&i\in\{1,\hdots,k\},\ j\in\{1,\hdots,m\},\\
q_i^{jl}(M_1,\dots,M_k)&:=e_j^\dagger M_i^\dagger M_i e_l-\delta_{jl}=0,\\
&i\in\{1,\hdots,k\},\ j,l\in\{1,\hdots,n\},
\end{split}
\end{align}
in $(M_1,\dots,M_k)\in \Pi_{i=1}^kM(m,n,\mathbb{C})$. Under the canonical identification $M(m,n,\mathbb{C})\simeq \R^{2nm}$, these can be considered as real algebraic equations in the $2knm$ variables $(M_1,\hdots,M_k)$.
\begin{lemma}\label{LEMvN}
Let $X\in H(\mathbb{C}^n)$ with $X\neq 0$. Imposing the equations \eqref{eqnblaaaa} on $\Pi_{i=1}^kM(m,n,\mathbb{C})$ decreases the dimension by at least $kn^2+k(m-1)$.
\end{lemma}
\begin{remark}
Regarding $X\in\mathcal{D}\subseteq H(\mathbb{C}^n)$ as an variable, the equations \eqref{eqnblaaaa} can be considered as equations on $\prod_{i=1}^kM(m,n,\mathbb{C})\times\mathcal{D}$. Then, Lemma \ref{LEMvN} implies that imposing the equations \eqref{eqnblaaaa} on $\prod_{i=1}^kM(m,n,\mathbb{C})\times\mathcal{D}$ decreases the dimension by at least $n^2+k(m-1)$ for every semi-algebraic set $\mathcal{D}\subseteq H(\mathbb{C}^n)$ with $0\notin\mathcal{D}$.
\end{remark}
Since the proof of this result is rather technical we relegate it to Section \ref{proof}. Lemma \ref{LEMvN} allows us to prove the main theorem of this section.
\begin{theorem}[Universality]\label{THMvN}
For $\mathcal{R}\subseteq\mathcal{S}(\mathbb{C}^n)$ a subset, let $\mathcal{D}$ be a semi-algebraic set that represents $\Delta(R)-\{0\}$. If $k(m-1)>\dim \mathcal{D}$, almost all measurement schemes $M\in\mathcal{M}_{1,k}^m(\mathbb{C}^n)$ are stably $\mathcal{R}$-complete.
\end{theorem}
\begin{remark}
Note that Theorem \ref{THMvN} reduces the problem of finding an $\mathcal{R}$-complete rank one POVM for some subset $\mathcal{R}\subseteq\mathcal{S}(\h)$ to finding a semi-algebraic subset $\mathcal{D}\subseteq H(\h)$ which represents $\Delta(\mathcal{R})-\{0\}$ and in this sense Theorem \ref{THMvN} guarantees the universality of rank one POVMs. Furthermore the quality of the result solely depends on the algebraic dimension of $\mathcal{D}$.
\end{remark}
The proof of this result can be found in Section \ref{proof}.

From this Theorem we directly obtain a Whitney type embedding result for rank one POVMs. Essentially, it is a direct consequence of the following lemma.
\begin{lemma}\label{lemdim}
Let $\mathcal{R}\subseteq\mathcal{S}(\h)$ be a semi-algebraic subset. Then $\dim(\Delta(\mathcal{R})-\{0\})\leq 2\dim\mathcal{R}$.
\end{lemma}
\begin{proof}
We can assume w.l.o.g that $\mathcal{R}$ is algebraic, because if not we can take its Zariski closure \footnote{The algebraic dimension is invariant under taking the Zariski closure, see Proposition 2.8.2 of \cite{bochnak1998real}}. Let $\text{Diag}(\mathcal{R}\times \mathcal{R}):=\{(X,Y)\in\mathcal{R}\times \mathcal{R}:X=Y\}$. Noting that $\text{Diag}(\mathcal{R}\times \mathcal{R})$ is an algebraic set,  $\mathcal{D}:=(\mathcal{R}\times \mathcal{R})-\text{Diag}(\mathcal{R}\times \mathcal{R})$ is quasi-algebraic. But the semi-algebraic map
\begin{align*}
\phi:\mathcal{D}&\to\Delta(\mathcal{R})-\{0\}\\
(X_1,X_2)&\mapsto X_1-X_2
\end{align*} 
is surjective, and thus $\dim(\Delta(\mathcal{R})-\{0\})\leq \mathcal{D}=2\dim\mathcal{R}$ by Theorem 2.8.8 of \cite{bochnak1998real}.
\end{proof}

\begin{corollary}\label{whitney}
Let $\mathcal{R}\subseteq\mathcal{S}(\mathbb{C}^n)$ be a subset. If $k(m-1)>2\dim \mathcal{R}$, almost all measurement schemes $M\in\mathcal{M}_{1,k}^m(\mathbb{C}^n)$ are stably $\mathcal{R}$-complete.
\end{corollary}
\begin{proof}
We can assume w.l.o.g. that $\mathcal{R}$ is algebraic because if not we can consider its Zariski closure. By the proof of Lemma \ref{lemdim}, $\Delta(\mathcal{R})-\{0\}$ is semi-algebraic and furthermore $\dim(\Delta(\mathcal{R})-\{0\})\leq 2\dim\mathcal{R}$. Finally, Theorem \ref{THMvN} with $\mathcal{D}=\Delta(\mathcal{R})-\{0\}$ concludes the proof.
\end{proof}

Two special cases of this Theorem may be of particular interest.
\begin{corollary}
Let $\mathcal{R}\subseteq\mathcal{S}(\mathbb{C}^n)$ be a subset. If $k(n-1)>2\dim \mathcal{R}$, almost all tuples of $k$ von Neumann measurement $M\in\mathcal{M}_{vN}^k(\mathbb{C}^n)$ are $\mathcal{R}$-complete.
\end{corollary}
\begin{proof}
This immediately follows from Corollary \ref{whitney} for $m=n$.
\end{proof}
\begin{corollary}\label{corrrankone}
Let $\mathcal{R}\subseteq\mathcal{S}(\mathbb{C}^n)$ be a subset. If $m-1>2\dim \mathcal{R}$, almost all rank one POVMs $M\in\mathcal{M}_{1}^m(\mathbb{C}^n)$ are stably $\mathcal{R}$-complete.
\end{corollary}
\begin{proof}
This immediately follows from Corollary \ref{whitney} for $k=1$.
\end{proof}
\begin{remark}
Effectively we have the bound $m-1>\max\{2\dim \mathcal{R},n-2\}$ which is due to the fact that a rank one POVM on $\mathbb{C}^n$ has to be at least $n$-dimensional. If we relax this to merely requiring the POVM to be projective this shortcoming can be avoided, i.e. for projective POVMs  $m-1=2\dim \mathcal{R}+1$ can be attained. This can be seen by modifying the proof of Lemma \ref{LEMvN}.
\end{remark}

\subsection*{Rank One POVMs for States of Bounded Rank and States of Fixed Spectrum}

In this section we improve the Whitney type bounds of Corollary \ref{whitney} for the cases in which the subset $\mathcal{R}\subseteq\mathcal{S}(\h)$ is given by the states of bounded rank or the states of fixed spectrum. The results we obtain in this section easily follow from theorem \ref{THMvN}. Let us note that all results of this section can be immediately transferred to measurement schemes which fulfil a universality property analogous to theorem \ref{THMvN}.

In the following, $r\in\{1,\hdots,[n/2]\}$. Denote by $S_r(\h)$ the states with rank at most $r$, i.e. $S_r(\h):=\{\varrho\in \mathcal{S}(\h): \text{rank}(\varrho)\leq r\}$. We write $S_r^n$ as shorthand for $S_r(\mathbb{C}^n)$.

In analogy to the proof of Theorem \ref{frame}, we first construct the set we use to represent $\Delta(S_r(\h))-\{0\}$ and determine its dimension.
\begin{lemma}\label{set}
The set $\mathcal{D}:=\{X\in\mathcal{P}_{2r}(\h):\ \text{tr}(X)=0,\ \text{tr}(X^2)=2\}$ is an algebraic set that represents $\Delta(S_r(\h))-\{0\}$ and $\dim\mathcal{D}=4r(\dim\h-r)-2$.
\end{lemma}
\begin{proof}
Note that $\mathcal{S}_r(\h)\subseteq \mathcal{P}_r(\h)$ and thus $\Delta(\mathcal{S}_r(\h))\subseteq \Delta(\mathcal{P}_r(\h))=\mathcal{P}_{2r}(\h)$. $\mathcal{P}_{2r}(\h)$ is algebraic by Lemma \ref{lemrank} and hence $\mathcal{P}_{2r}(\h)-\{0\}$ represents $\Delta(\mathcal{S}_r(\h))-\{0\}$. In fact $\Delta(\mathcal{S}_r(\h))-\{0\}$ can be represented by a smaller set. Namely one can consider set $\mathcal{D}:=\{X\in\mathcal{P}_{2r}(\h):\ \|X\|_2^2=tr(X^2)=1,\ tr(X)=0\}$. Note that $\mathcal{D}$ is algebraic by Corollary \ref{corrank} and that $0\notin\mathcal{D}$. The equation $tr(X)=0$ just considers the fact that states have unit trace. Next consider a measurement scheme $M$ and $X\in\Delta(\mathcal{S}_r(\h))-\{0\}$ such that $h_M(X)=0$. Then, there is $X^\prime:=\frac{X}{\|X\|_2}\in\mathcal{D}$ such that $h_M(X^\prime)=0$. Hence $\mathcal{D}$ indeed represents $\Delta(\mathcal{S}_r(\h))-\{0\}$. Finally, by Corollary \ref{corrank}, we have $\dim(\mathcal{D})=\dim(\mathcal{P}_{2r}(\h))-2=4r(n-r)-2$.
\end{proof}
\begin{theorem}\label{thmvN}
If $k(m-1)\ge 4r(n-r)-1$, almost all measurement schemes $M\in M_{1,k}^m(\mathbb{C}^n)$ are stably $\mathcal{S}_r^n$-complete.
\end{theorem}
\begin{proof}
Using the set of Lemma \ref{set} to represent $\Delta(\mathcal{R})-\{0\}$, the result follows directly form Theorem \ref{THMvN}.
\end{proof}

As explained in Section IV.A of \cite{kech}, the lower bounds on the immersion dimension of complex flag manifolds of \cite{walgenbach2001lower} transfer to lower bounds on the dimension of $\mathcal{S}_r(\h)$-complete POVMs. In addition, the discussion following this explanation suggests that the upper bound on $m$ we obtain here is close to optimal.

Next, let us state some corollaries of this theorem.
\begin{corollary}\label{corrvN}
If $k(n-1) \ge 4r(n-r)-1$, almost all tuples of $k$ von Neumann measurements $M\in\mathcal{M}_{\text{vN}}^k(\mathbb{C}^n)$ are stably $\mathcal{S}_r^n$-complete.
\end{corollary}
\begin{proof}
This follows from Theorem \ref{thmvN} for $m=n$.
\end{proof}
For $r=1$ this reproduces the main result of \cite{mondragon2013determination}. In Table \ref{constantspectrumbounds} you can see how this result compares to the lower bounds of \cite{walgenbach2001lower} for some explicit scenarios.
 \begin{table}[h]
  \setlength{\tabcolsep}{10mm}
 \begin{tabular}{lcccccccccccc}
l\textbackslash k   &         2 & 3         & 4\\
\\
                 5  & 6/7\\
  \\
                 6  & 6/7\\ 
  \\
                 7  & 7/7 & 9/10\\
 \\
                 8  & 7/7 & 9/10\\
 \\
                 9  & 7/8 & 9/10  & 12/12\\
 \\
                 10 & 7/8 & 10/10 & 12/13\\

 \end{tabular}
 \caption{Lower bounds on the minimal number of von Neumann measurements necessary to discriminate any two quantum states of rank at most $k$ from \cite{walgenbach2001lower} for $\mathcal{S}_k^{k+l}$./ Upper bounds on the minimal number of von Neumann measurements necessary to discriminate any two quantum states of rank at most $k$ from Corollary \ref{corrvN} for $\mathcal{S}_k^{k+l}$.}
\label{constantspectrumbounds}
\end{table}

\begin{corollary}\label{corpars}
If $m-1 \ge 4r(n-r)-1$, almost all rank one POVM $P\in\mathcal{M}_1^m(\mathbb{C}^n)$ are stably $\mathcal{S}_r^n$-complete.
\end{corollary}
\begin{proof}
This follows from \ref{thmvN} for $k=1$.
\end{proof}

Finally we consider states of fixed spectrum. Let $s$ \footnote{A spectrum on $\mathbb{C}^n$ is a multiset of $n$ increasingly ordered positive real numbers that sum up to one. We call the elements of $s$ eigenvalues.}  
be a spectrum on $\mathbb{C}^n$ and denote by $\mathcal{S}_s^n\subseteq\mathcal{S}(\mathbb{C}^n)$ the states with spectrum $s$.

\begin{corollary}
Let $s$ be a spectrum on $\mathbb{C}^n$ such that the highest multiplicity of an eigenvalue in $s$ is $n-r$. Then, if $k(n-1)\ge 4r(n-r)-1$, almost all tuples of $k$ von Neumann measurements $M\in\mathcal{M}_{\text{vN}}^k(\mathbb{C}^n)$ are stably $\mathcal{S}_s^n$-complete.
\end{corollary}
\begin{proof}
This follows directly from Theorem \ref{thmvN} for $m=n$ noting that $\Delta(\mathcal{S}_s^n)-\{0\}$ can be represented by the set of Lemma \ref{set} \footnote{For more details see Lemma IV.3 of \cite{kech}.}.
\end{proof}

\begin{corollary}
Let $s$ be a spectrum on $\mathbb{C}^n$ such that the highest multiplicity of an eigenvalue in $s$ is $n-r$. Then, if $m-1\ge 4r(n-r)-1$, almost all POVMs $P\in M_{1}^m(\mathbb{C}^n)$ are stably $\mathcal{S}_s^n$-complete.
\end{corollary}
\begin{proof}
This follows directly from Theorem \ref{thmvN} for $k=1$ noting that $\Delta(\mathcal{S}_s^n)-\{0\}$ can be represented by the set of Lemma \ref{set}.
\end{proof}

\section{ Quantum Tomography with Local Observables}
 
In this section we address the problem of reconstructing states of multipartite systems from the expectation values of local observables.

Let $\h=\bigotimes_{i=1}^k\mathbb{C}^{n_i}$ and let $n:=\prod_{i=1}^kn_i$. We define the set $H_{loc}(\h)$ of local observables on $\h$ by
\begin{align*}
H_{loc}(\h):=\{O_1\otimes \hdots\otimes O_k:O_i\in SH(\mathbb{C}^{n_i})\}\subseteq H(\h).
\end{align*}
Just like a POVM, a tuple of observables $O:=(O_1,\hdots,O_m)\in H(\h)^m$, induces a linear map $h_O:H(\h)\to \R^m,\ X\mapsto (\text{tr}(O_1X),\hdots,\text{tr}(O_mX))$ and hence Definition \ref{defcomplete} and \ref{defstab} naturally generalize to finite tuples of observables. 

The following theorem is the analogue of Theorem \ref{THMvN} and it is the main result of this section.
\begin{theorem}\label{THMlocal}(Universality)
For $\mathcal{R}\subseteq \mathcal{S}(\h)$ let $\mathcal{D}\subseteq H(\h)$ be a semi-algebraic set that represents $\Delta(\mathcal{R})-\{0\}$.  If $m>\dim\mathcal{D}$, almost all $O\in H_{loc}(\h)^m$ are stably $\mathcal{R}$-complete.
\end{theorem}
The proof of this Theorem is given in Section \ref{proof}. 

Again, we directly obtain a Whitney type embedding result for subsets $\mathcal{R}\subseteq\mathcal{S}(\h)$ if the measurement consists of determining expectation values of local observables.

\begin{corollary}\label{whitney2}
Let $\mathcal{R}\subseteq\mathcal{S}(\h)$ be a subset. If $m>2\dim \mathcal{R}$, almost all $O\in H_{loc}(\h)^m$  are stably $\mathcal{R}$-complete. 
\end{corollary}
\begin{proof}
We can assume w.l.o.g. that $\mathcal{R}$ is algebraic because if not we can consider its Zariski closure. By the proof of Lemma \ref{lemdim}, $\Delta(\mathcal{R})-\{0\}$ is semi-algebraic and $\dim(\Delta(\mathcal{R})-\{0\})\leq 2\dim\mathcal{R}$. Finally, Theorem \ref{THMlocal} concludes the proof.
\end{proof}

Just like in the case of rank one POVMs also this measurement scheme applies to the problem of discriminating states of bounded rank or states of fixed spectrum.
\begin{corollary}\label{ranklocal}
If $m \ge 4r(n-r)-1$, almost all $O\in H_{loc}(\h)^m$ are stably $\mathcal{S}_r(\h)$-complete.
\end{corollary}
\begin{proof}
Let $\mathcal{D}$ be the quasi-algebraic set of Lemma \ref{set}. Then the result follows directly from Theorem
\ref{THMlocal}. 
\end{proof}
\begin{corollary}
Let $s$ be a spectrum on $\h$ such that the highest multiplicity of an eigenvalue in $s$ is $n-r$. If $m \ge 4r(n-r)-1$, almost all $O\in H_{loc}(\h)^m$ are stably $\mathcal{S}_s^n$-complete.
\end{corollary}
\begin{proof}
This follows directly from Corollary \ref{ranklocal} noting that the set of Lemma \ref{set} represents  $\Delta(\mathcal{S}_s^n)-\{0\}$.
\end{proof}

Finally, let us apply Theorem \ref{THMlocal} to local Pauli observables on qubit systems. Let $\mathcal{H}=\bigotimes_{i=1}^d\mathbb{C}^2$. The set of local Pauli observables $H_\sigma(\h)$ on  $\h$ is given by
\begin{align*}
H_\sigma (\h):=\{\sigma_1\otimes \hdots\otimes\sigma_d:\sigma_i\in SH(\mathbb{C}^{n_i})_0\}
\end{align*}
where $H(\mathbb{C}^{n_i})_0:=\{X\in H(\mathbb{C}^{n_i})_0:\text{tr}(X)=0\}$ is the real vector space of traceless hermitian $n_i\times n_i$ matrices and $SH(\mathbb{C}^{n_i})_0:=\{X\in H(\mathbb{C}^{n_i})_0:\|X\|_2=1\}$ is the unit sphere in $H(\mathbb{C}^{n_i})_0$.
\begin{corollary}
If $m \ge 4r(2^d-r)-1$, almost all $O\in H_{\sigma}(\h)^m$ are stably $\mathcal{S}_r(\h)$-complete.
\end{corollary}
\begin{proof}
Theorem \ref{THMlocal} also holds for $H_{\sigma}(\h)$ \footnote{See the remark after proof of Lemma \ref{LEMlocal}.}. The remainder of the proof is then along the lines of the proof of Corollary \ref{ranklocal}.
\end{proof}

\section{Technical Results}\label{proof}

\subsection*{Proof of Lemma \ref{LEMvN}}

Before giving the proof of Lemma \ref{LEMvN} let us first explain the methods we use to compute the dimension of the relevant algebraic set.

We take advantage of the fact that the dimension of an algebraic set $V$ is given by the dimension of the tangent space  at non-singular points of $V$(see Definition 3.3.3 of \cite{bochnak1998real}). Let us make this more precise: Let $\mathbb{R}[x_1,\hdots,x_n]$ be the ring of real polynomials in $n$ variables and denote by $dp$ the differential of a real polynomial $p\in\mathbb{R}[x_1,\hdots,x_n]$, i.e. $dp(y)=\sum_{i=1}^{n}\frac{\partial p}{\partial x_i}|_ydx_i$. Let $V_I$ be the real common zero locus of a set of real polynomials $I:=\{p_1,\hdots,p_m\}\subseteq \mathbb{R}[x_1,\hdots,x_n]$. For all $x\in V_I$, 
\begin{align}\label{eqnsrank}
\sum_{i=1}^{m}\alpha_i dp_i(x)=0
\end{align}
gives a system of linear equations in $\alpha_1,\hdots,\alpha_m\in\mathbb{R}$. In the following we mainly use the following facts:
\begin{enumerate}
\item[1.]The rank of the system of linear equations \eqref{eqnsrank} at a non-singular point of $V_I$ is given by $n-d$ where $d$ is the dimension of $V_I$ \footnote{See Definition 3.3.4 and Proposition 3.3.10 of \cite{bochnak1998real}.}.
\item[2.] The non-singular points of $V_I$ are an algebraic subset of dimension less than $d$ by Proposition 3.3.14 of \cite{bochnak1998real}.
\end{enumerate}.

By computing these systems of linear equations, we prove that for a given non-zero $X\in H(\mathbb{C}^n)$, imposing the equations \eqref{eqnblaaaa} on $\Pi_{i=1}^kM(m,n,\mathbb{C})$ decreases the dimension by at least $n^2+k(m-1)$.

First, let us state a lemma which allows us to efficiently compute the systems of linear equations for the equations \eqref{eqnblaaaa}. Let $A\in M(s,m,\R)$, $C\in M(m,t,\R)$, $B\in H(\mathbb{C}^n)$. Furthermore, identify $M(m,n,\mathbb{C})$ with $\mathbb{R}^{2mn}$ via the canonical map $\iota:M(m,n,\mathbb{C})\to\mathbb{R}^{2mn},\ Y\mapsto (\text{Re}(Y),\text{Im}(Y))$ . Then the equations 
\begin{gather*}
p^I_{lo}(Y):=\text{Im}(AYBY^\dagger C)_{lo}=0,\ \ p^R_{lo}(Y):=\text{Re}(AYBY^\dagger C)_{lo}=0,\\ l\in\{0,\hdots,s\},o\in\{0,\hdots,t\},
\end{gather*}
can be considered as real algebraic equations in the variables $y^R_{jk}:=(\text{Re}(Y))_{jk}, y^I_{jk}:=(\text{Im}(Y))_{jk},\ j\in\{0,\hdots,m\},k\in\{0,\hdots,n\}$.

\begin{lemma}\label{lemlgs}
Let $Y\in M(m,n,\mathbb{C})$ be such that $AYBY^\dagger C=0$. Then, the system of linear equations
\begin{align*}
L(Y):=\sum_{l=1}^{s}\sum_{o=1}^{t}\left(\alpha^R_{lo}dp^R_{lo}(Y)+\alpha^I_{lo}dp^I_{lo}(Y)\right)=0
\end{align*}
in $\alpha^R_{lo}\in\R,\ \alpha^I_{lo}\in\R$ is equivalent to $A^TM_\alpha C^TYB+CM_\alpha^\dagger AYB=0$ where $(M_\alpha)_{lo}:=\alpha^{R}_{lo}+i\alpha^I_{lo},\ l\in\{0,\hdots,s\},o\in\{0,\hdots,t\}$.
\end{lemma}

\begin{proof}
Let
\begin{align*}
L_{jk}=(\partial_{y^R_{jk}}-i\partial_{y^I_{jk}}) \sum_{l=1}^s\sum_{o=1}^t \left(\alpha_{lo}^R\ p ^R_{lo}+\alpha_{lo}^I\ p^I_{lo}\right),\ j\in\{1,\hdots,m\},k\in\{1,\hdots,n\}.
\end{align*}
Then the system of linear equations $\{L_{jk}(Y)=0\}_{j\in\{1,\hdots,m\},k\in\{1,\hdots,n\}}$ is equivalent to $L(Y)=0$ since
\begin{small}
\begin{align*}
L&=\sum_{l=1}^s\sum_{o=1}^t \left(\alpha_{lo}^R\sum_{j=1}^m\sum_{k=1}^n\left((\partial_{y^R_{jk}}p^R_{lo})dy^R_{jk}
   +(\partial_{y^I_{jk}}p^R_{lo})dy^I_{jk}\right)+\alpha_{lo}^I\sum_{j=1}^m\sum_{k=1}^n
   \left((\partial_{y^R_{jk}}p^I_{lo})dy^R_{jk}+(\partial_{y^I_{jk}}p^I_{lo})dy^I_{jk}\right)\right)\\
&=\sum_{j=1}^m\sum_{k=1}^n\left(\left(\partial_{y^R_{jk}}\sum_{l=1}^s\sum_{o=1}^t   \left(\alpha_{lo}^R\ p^R_{lo}+\alpha_{lo}^I\ p^I_{lo}\right)\right)dy^R_{jk}+\left(\partial_{y^I_{jk}}\sum_{l=1}^s\sum_{o=1}^t \left(\alpha_{lo}^R\ p^R_{lo}+\alpha_{lo}^I\ p^I_{lo}\right)\right)dy^I_{jk}\right)\\
  &=\sum_{j=1}^m\sum_{k=1}^n\left(\text{Re}(L_{jk})dy_{jk}^R-\text{Im}(L_{jk})dy_{jk}^I\right).
\end{align*}
\end{small}

Let $\partial_{y_{jk}}=\partial_{y_{jk}^R}-i\partial_{y_{jk}^I}$ and note that $\partial_{y_{jk}}Y_{lm}=2\delta_{jl}\delta_{km}$,  $\partial_{y_{jk}}Y^*_{lm}=0$. Then,
\begin{small}
\begin{align*}
L_{jk}(Y)&=(\partial_{y_{jk}^R}-i\partial_{y_{jk}^I}) \sum_{l=1}^s\sum_{o=1}^t \left( \alpha_{lo}^R \frac{1}{2}(AYBY^\dagger C+AY^*B^*Y^T C)_{lo}+\alpha_{lo}^I\frac{1}{2i}(AYBY^\dagger C-AY^*B^*Y^T C)_{lo}\right)\\
&=\frac{1}{2}\partial_{y_{jk}}\sum_{l=1}^s\sum_{o=1}^t \left((M_\alpha^*)_{lo} (AYBY^\dagger C)_{lo}+(M_\alpha)_{lo}(AY^*B^*Y^T C)_{lo}\right)\\
&=\sum_{l=1}^s\sum_{o=1}^t\sum_{p=1}^m\sum_{q=1}^n\left((M_\alpha^*)_{lo}A_{lp}\delta_{pj}\delta_{qk}(BY^\dagger C)_{qo}+(M_\alpha)_{lo}(AY^*B^*)_{lq}\delta_{qk}\delta_{pj}C_{po}\right)\\
&=(A^TM_\alpha^*C^TY^*B^T+CM_\alpha^T AY^*B^*)_{jk}\\
&=(A^TM_\alpha C^TYB+CM_\alpha^\dagger AYB)^*_{jk}.
\end{align*}
\end{small}
Hence $L(Y)=0$ is equivalent to
$A^TM_\alpha C^TYB+CM_\alpha^\dagger AYB=0$.
\end{proof}
Under the identification $M(m,n,\mathbb{C})\simeq \R^{2mn}$ given by the map $\iota$ defined above, also the equations 
\begin{gather*}
r^R_{lo}(Y):=\text{Re}(Y^{\dagger}Y)_{lo}-\delta_{lo}=0,\ r^I_{lo}(Y):=\text{Im}(Y^{\dagger}Y)_{lo}=0,\\ l,o\in\{1,\hdots,n\},
\end{gather*}
can be considered as real algebraic equations in the variables $y^R_{jk}:=(\text{Re}(Y))_{jk}, y^I_{jk}:=(\text{Im}(Y))_{jk},\ j\in\{0,\hdots,m\},k\in\{0,\hdots,n\}$.
\begin{corollary}\label{corlgs}
Let $Y\in M(m,n,\mathbb{C})$ be such that $Y^{\dagger}Y-\id_n=0$. Then, the system of linear equations
\begin{align*}
L(Y):=\sum_{l,o=1}^{n}\left(\gamma^R_{lo}dr^R_{lo}(Y)+\gamma^I_{lo}dr^I_{lo}(Y)\right)=0
\end{align*}
in $\gamma^R_{lo}\in\R,\ \gamma^I_{lo}\in\R$ is equivalent to $Y(M_\gamma+M_\gamma^\dagger)=0$ where $(M_\gamma)_{lo}:=\gamma^{R}_{lo}+i\gamma^I_{lo},\ l,o\in\{1,\cdots,n\}$.
\end{corollary}
\begin{proof}
The proof of this result can be obtained by going along the lines of the proof of Lemma \ref{lemlgs}, so we just give the calculation that differs: $L(Y)=0$ is equivalent to $\{L_{jk}(Y)=0\}_{j\in\{1,\cdots,m\},k\in\{1,\cdots,n\}}$ where
\begin{small}
\begin{align*}
L_{jk}(Y)&=(\partial_{y_{jk}^R}-i\partial_{y_{jk}^I}) \sum_{l,o=1}^n \left( \gamma_{lo}^R \frac{1}{2}(Y^{\dagger}Y+Y^{T}Y^*)_{lo}+\gamma_{lo}^I\frac{1}{2i}(Y^{\dagger}Y-Y^{T}Y^*)_{lo}\right)\\
&=\frac{1}{2}\partial_{y_{jk}}\sum_{l,o=1}^n \left((M_\gamma^*)_{lo} (Y^{\dagger}Y)_{lo}+(M_\gamma)_{lo}(Y^{T}Y^*)_{lo}\right)\\
&=\sum_{l,o=1}^n\sum_{p=1}^m\left((M_\gamma^*)_{lo}\delta_{ko}\delta_{jp}(Y^\dagger)_{lp}+(M_\gamma)_{lo}(Y^*)_{po}\delta_{lk}\delta_{pj}\right)\\
&=(Y^*M_\gamma^*+Y^*M_\gamma^T)_{jk}\\
&=(YM_\gamma+YM_\gamma^\dagger)^*_{jk}.
\end{align*}
\end{small}
Hence $L(Y)=0$ is equivalent to $Y(M_\gamma+M_\gamma^\dagger)=0$. 
\end{proof}
\begin{remark}
Note that combining the equations of Lemma \ref{lemlgs} and Corollary \ref{corlgs} yields the system of linear equations $Y(M_\gamma+M_\gamma^\dagger)+A^TM_\alpha C^TYB+CM_\alpha^\dagger AYB=0$ (see equations \ref{eqnsrank}).
\end{remark}

Let us now give the proof of Lemma \ref{LEMvN}.
\begin{proof}
For a given non-zero $X\in H(\mathbb{C}^n)$ and $i\in\{1,\hdots,k\}$, consider the following equations in $(M_1,\dots,M_k)\in \prod_{i=1}^kM(m,n,\mathbb{C})$:
\begin{align*}
p_i^j(M_1,\hdots,M_k):=\text{tr}(M_i^\dagger e_j e_j^\dagger M_iX)=e_j^\dagger M_iXM_i^\dagger e_j&=0,\ \ j\in\{1,\hdots,m\},
\end{align*}
and
\begin{align*}
q_i^{jl}(M_1,\hdots,M_k):=(M_{i}^\dagger M_{i})_{jl}-\delta_{jl}=0, \ \ j,l\in\{1,\hdots,n\}.
\end{align*}
Under the canonical identification $\prod_{i=1}^kM(m,n,\mathbb{C})\simeq \R^{2knm}$, these equations can be regarded as real algebraic equations in $2knm$ variables. Let $I_i:=\{p_i^j\}_{j\in\{1,\hdots,m\}}$ and $J_i:=\{q_i^{jl}\}_{j,l\in\{1,\hdots,n\}}$.

We have to show that the dimension of the real common zero locus of the equations $K_k=\bigcup_{i=1}^k I_i\cup J_i $ is at most $2kmn-kn^2-k(m-1)$. Denote by $\iota_1: M(m,n,\mathbb{C})\to \Pi_{i=1}^k M(m,n,\mathbb{C}),\ M\mapsto (M,0,\hdots)$ the inclusion in the first factor and let $\pi_i:\Pi_{i=1}^k M(m,n,\mathbb{C})\to M(m,n,\mathbb{C}) ,\ (M_1,\hdots,M_i,\hdots,M_k)\mapsto M_i$ be the projection on the $i$-th factor. Then we find $I_i\cup J_i=(J_1\cup I_1)\circ \iota_1\circ \pi_i$, where $(J_1\cup I_1)\circ \iota_1\circ \pi_i:=\{p\circ \iota_1\circ \pi_i:p\in J_1\cup I_1\}$. Thus, we conclude that $V_{K_k}\simeq \prod_{i=1}^k V_{K_1}$ and it suffices to reduce to $k=1$. We stick to the notation introduced in the beginning of this section and denote the algebraic set obtained from $M(m,n,\mathbb{C})$ by imposing the equations $I:=I_1$ and $J:=J_1$ by $V_{I\cup J}$.

Let us now determine the system of linear equations $L$ associated to $I\cup J$ at $U\in V_{I\cup J}$. The contribution of the $j$-th equation of $I$ to $L$ is obtained from Lemma \ref{lemlgs} by choosing $A=e_j^\dagger$, $B=X$, $C=e_j$, $Y=U$ and thus the contribution of $I$ is given by
\begin{align*}
\sum_{j=1}^m\alpha_j^R e_j e_j^\dagger UX,\ \alpha_j^R\in\R.
\end{align*}
Similarly, by Corollary \ref{corlgs}, the contribution of $J$ to $L$ is given by,
\begin{align*}
U(M_\gamma+M_\gamma^\dagger)
\end{align*}
where $(M_\gamma)_{jk}:=\gamma_{jk}^R+i\gamma_{jk}^I, i,j\in\{1,\hdots n\}, \gamma_{jk}^R,\gamma_{jk}^I\in\R$. Note that this just gives conditions on the hermitian part of $M_\gamma$ and define $\Gamma\in H(\mathbb{C}^n)$ by $\Gamma:=M_\gamma+M_\gamma^\dagger$.

Combining these two parts, the system of linear equations associated to the equations $I\cup J$ at $U\in V_{I\cup J}$ is equivalent to the following system of linear equations in $\alpha_1,\hdots,\alpha_m\in\R$ and $\gamma_{kj}^R\in \R, \gamma_{kj}^I\in \R,\ k,j\in\{1,\hdots,n\},$ 
\begin{align}\label{lgs1}
U\Gamma  + D_\alpha UX=0
\end{align}
where $D_{\alpha}=\sum_{j=1}^m\alpha_je_j e_j^\dagger$. Observing that $\Gamma$ is uniquely determined by the equations \eqref{lgs1}, the rank of \eqref{lgs1} is at least $n^2$ and we can reduce to the anti-hermitian part of \eqref{lgs1} to find the remaining $m-1$ independent equations:
\begin{align}\label{lgs}
0=U\Gamma U^\dagger + D_\alpha U XU^\dagger  -\left(U\Gamma U^\dagger +D_\alpha U XU^\dagger \right)^\dagger=-[U XU^\dagger,D_\alpha].
\end{align}

{ Next we study the commutator $[U XU^\dagger,D_\alpha]$ in detail. As $X$ is an arbitrary hermitian matrix, we have to carefully consider all possible combinations of eigenspaces, or more precisely eigenspace degeneracies, $X$ could have.

In order to achieve this, let us begin with the following example, which will be the starting point for the decomposition of $X$:} Let $M$ be a subset of $\{1,\hdots,m\}$ and define the diagonal projection $D_M\in M(m,\mathbb{R})$ by $e_i^\dagger D_Me_j:=\delta_{i,j}\delta_{j,M}$, where $\delta_{j,M}=1\text{ for }j\in M$ and $0$ else. The following observation is the crucial idea for the remainder of the proof: If $[U XU^\dagger,D_M]\neq 0$ for all proper subsets $M$ of $\{1,\hdots,m\}$ then $m-1$ of the operators $\{[U XU^\dagger,D_{\{i\}}]\}_{i\in\{1,\hdots,m\}}$ are linearly independent. To show this, assume that there are $a_j\in \R,\ j\in\{1,\hdots,m\}$, with $a_k\neq a_l$ for some $k,l$ such that $\sum_{j=1}^m a_j[U XU^\dagger,D_{\{j\}}]=0$. Since the commutativity of hermitian matrices is determined solely by their eigenspaces, we deduce $[U XU^\dagger,D_{E}]=0$, where $E:=\{j\in\{1,\hdots,m\}:a_j=a_k\}$. But this is a contradiction since $E$ is a proper subset of $\{1,\cdots,m\}$. Hence, the only solution is $a_1=a_2=\cdots=a_m$ and this proves the claim.
Thus, in this case we conclude  that the solution of the system of linear equations \eqref{lgs} is given by  $\alpha_1=\hdots=\alpha_m$ and hence there are $m-1$ linearly independent equations. 

Next, we decompose $V_{I\cup J}$ into quasi-algebraic subsets for which the argument we just gave can be applied\footnote{{By means of this decomposition we can separately consider all possible eigenspace degeneracies of $X$.}}. Let $P[m]$ be the set of partitions of $\{1,\hdots,m\}$. We say that a subset $S\subseteq\{1,\hdots,m\}$ is subordinate to a partition $P\in P[m]$ if there is $M\in P$ such that $S$ is a proper subset of $M$. For given $P\in P[m]$, define the quasi-algebraic set $W_P$ to be the set of $U\in V_{I\cup J}$ such that
\begin{align}\label{eqcommu}
[D_M,U XU^\dagger ]=0,\ \forall M\in P,
\end{align}
and
\begin{align*}
[D_N,U XU^\dagger ]\neq 0,\ \forall N\subseteq\{1,\hdots,m\}\ \text{subordiante to }P.
\end{align*}
The set $V_{I\cup J}$ can clearly be decomposed into the sets $W_P$:
\begin{align*}
V_{I\cup J}=\bigcup_{P\in P[m]}W_P.
\end{align*}

Having already checked that $2mn-\dim W_P= m-1+n^2$ if $P$ is the trivial partition, we conclude the proof by showing that $2mn-\dim W_P\ge m-1+n^2$ for all non-trivial $P\in P[m]$ \footnote{Note that, depending on the choice of $X$, many of the $W_P$ might be empty. If $X=\id_n$, $n=m$ for instance, all $W_P$ would be empty.}. In order to prove this, we first show that the rank of the system of linear equations associated to $W_P$ is at least $n^2+m-1$ for all points in $W_P$.

Let $P=\{M_1,\hdots,M_l,M_{l+1}\}\in P[m]$ be an arbitrary non-trivial partition. Choosing $A=D_{M_j}$, $B=X$, $C=\text{id}_m$ and $Y=U$ in Lemma \ref{lemlgs} yields
\begin{align*}
D_{M_j}M_{\beta_j}UX +M_{\beta_j}^\dagger D_{M_j}UX,
\end{align*}
where $M_{\beta_j}\in M(m,\mathbb{C})$ with $(M_{\beta_j})_{lo}:=\beta^R_{j;lo}+i\beta^I_{j;lo},\ l,o\in\{1,\hdots,m\},\ \beta^R_{j;lo},\beta^I_{j;lo}\in\R$
and similarly with the roles of $A$ and $C$ exchanged. Thus, equation \eqref{eqcommu} for $M_j$ gives the following contribution to the system of linear equations associated to $W_P$ at $U\in W_P$:
\begin{align*}
[D_{M_j},M_{\beta_j}-M_{\beta_j}^\dagger]UX.
\end{align*}

Thus, the system of linear equations associated to $W_P$ at $U\in W_P$ is equivalent to the following system of linear equations in $\alpha_1,\hdots,\alpha_m\in\R$, $\gamma_{kj}^R\in \R, \gamma_{kj}^I\in \R,\ k,j\in\{1,\hdots,n\},$ and $\beta^R_{j;lo}\in\R,\beta^I_{j;lo}\in\R,\ j\in\{1,\hdots,l+1\},l,o\in\{1,\hdots,m\}$:
\begin{align*}
U\Gamma  + D_\alpha UX+\sum_{k=1}^{l+1}[D_{M_k},M_{\beta_k}-M_{\beta_k}^\dagger]UX=0.
\end{align*}
Again, we can eliminate $\Gamma$ by reducing to the anti-hermitian part to obtain
\begin{small}
\begin{align} \label{lgs2}
U\Gamma U^\dagger+D_{\alpha} U XU^\dagger +&\sum_{k=1}^{l+1}[D_{M_k},M^H_{\beta_k}]UXU^\dagger-\left(U\Gamma U^\dagger+D_{\alpha} U XU^\dagger +\sum_{k=1}^{l+1}[D_{M_k},M_{\beta_k}^H]UXU^\dagger\right)^\dagger \nonumber\\
&\Leftrightarrow [U XU^\dagger ,D_{\alpha}]+\sum_{k=1}^{l+1}[U XU^\dagger ,[M_{\beta_k}^H,D_{M_k}]]=0, 
\end{align}
\end{small}
where $M_{\beta_j}^H$ is the anti-hermitian $m\times m$ matrix defined by $M_{\beta_j}^H:=M_{\beta_j}-M_{\beta_j}^\dagger$. 

Conjugating with $D_{M_j}$ yields
\begin{align*}
[UXU^\dagger,D_{M_j}D_{\alpha}]=0,
\end{align*}
where we used $[UXU^\dagger,D_{M_j}]=0$ together with $D_{M_j}[M_{\beta_k}^H,D_{M_k}]D_{M_j}=D_{M_j}M_{\beta_j}^HD_{M_j}-D_{M_j}M_{\beta_j}^HD_{M_j}=0$. By construction of $W_P$, we have $[UXU^\dagger,D_{M_j}D_M]\neq 0$ for all proper subsets $M\subseteq M_j$. Since the commutativity of hermitian matrices is solely determined by their eigenspaces  we conclude just like in the case of the trivial partition that $D_{M_j}D_{\alpha}\propto D_{M_j}$ for all $j\in\{1,\hdots,l+1\}$ \footnote{In particular, note that if $D_\alpha$ solves the system of linear equations \eqref{lgs2} we have $[U XU^\dagger,D_{\alpha}]=0$.}. Thus, if there is $U\in W_P$, the rank of \eqref{lgs2} at $U$ is at least $n^2+m-l-1$. 

To find the remaining $l$ independent equations consider the remaining equations
\begin{align*}
\sum_{j=1}^{l+1}[U XU^\dagger,[M^H_{\beta_j},D_{M_j}]]=0.
\end{align*}

There is $i\in\{1,\hdots,l+1\}$ with $D_{M_i}U XU^\dagger\neq 0$ because otherwise we would conclude that $U XU^\dagger= 0$ which is a contradiction since $U\in U(m,n)$ and $X\neq 0$ by assumption. Multiplying by $D_{M_k}$ from the left and $D_{M_i}$ from the right yields
\begin{align*}
&\sum_{j=1}^{l}D_{M_k}[U XU^\dagger,[M^H_{\beta_j}, D_{M_j}]]D_{M_i}=0\\
\Leftrightarrow&\sum_{j=1}^{l}[U XU^\dagger ,D_{M_k}[M^H_{\beta_j}, D_{M_j}]]D_{M_i}=0\\
\Leftrightarrow&[U XU^\dagger ,D_{M_k}(M^H_{\beta_{i}}-M^H_{\beta_k}) D_{M_i}]=0.
\end{align*}
For each $k\in\{1,\hdots,l+1\}-\{i\}$ this gives at least one equation on $M^H_{\beta_k}$: First, assume $|M_{i}|=1$. Then there is $q\in\{1,\dots,m\}$ such that $M_i=\{q\}$. Furthermore, since 
\begin{align*}
0\neq D_{M_i}U XU^\dagger=D_{M_i}U XU^\dagger D_{M_i}=(e_q^dagger U XU^\dagger q)qq^\dagger,
\end{align*}
we conclude that $e_q^\dagger U XU^\dagger q\neq 0$. But this is a contradiction to the $q$-th equation of $I$.  

Hence we can assume $|M_{i}|\geq 2$. By construction of $W_P$ there is an eigenvector $v_k\neq 0$ of $U XU^\dagger$ in the range of $D_{M_k}$ with eigenvalue $\lambda_k$ and a eigenvector $v_{i}\neq 0$ of $U XU^\dagger$ in the range of $D_{M_{i}}$ with eigenvalue $\lambda_{i}$. Since we assumed $|M_{i}|\geq 2$, by construction of $W_P$, $U XU^\dagger$ has at least two eigenvectors in the range of $D_{M_{i}}$ with different eigenvalues because otherwise there would be a proper subset of $N\subseteq M_i$ such that $[UXU^\dagger,D_N]=0$. Thus, we can choose $\lambda_{i}$ such that $\lambda_{i}\neq \lambda_k$. We then find
\begin{align*}
& v_k^\dagger[U XU^\dagger ,D_{M_k}(M^H_{\beta_{i}}-M^H_{\beta_k}) D_{M_i}]v_i=0\\
\Leftrightarrow& v_k^\dagger(M^H_{\beta_{i}}-M^H_{\beta_k})v_i(\lambda_k-\lambda_i)=0\\
\Leftrightarrow& v_k^\dagger M^H_{\beta_{i}}v_i- v_k^\dagger M^H_{\beta_k}v_i=0.
\end{align*}
But this clearly gives a non-trivial condition on $M^H_{\beta_k}$ since $M^H_{\beta_k}=M_{\beta_j}-M_{\beta_j}^\dagger$. Thus we conclude that, if there is $U\in W_P$, the rank of \eqref{lgs2} at $U$ is at least $m-l-1+l=m-1$ and hence the rank of the system of linear equations associated to $W_P$ at $U$ is at least $n^2+m-1$. But if $W_P$ is non-empty, it does contain a non singular-point by Proposition 3.3.14 of \cite{bochnak1998real}. And thus the rank of the system of linear equations associated to $W_P$ at this non-singular point is at least $n^2+m-1$. Hence $2kmn-\dim W_P\ge n^2+m-1$ by Proposition 3.3.10 of \cite{bochnak1998real}.

\end{proof}
\subsection*{Proof of Theorem \ref{THMvN}}\label{bla}
\begin{proof}
Let $\psi$ be the map defined in \eqref{psi}. We can assume that $\mathcal{D}$ is a closed subset of $SH(\mathbb{C}^n)$ because if not we can replace it by the closure of $\psi(\mathcal{D})$ without increasing its dimension \footnote{See remark after Lemma \ref{stab} for more details.}. Let $\tilde{\mathcal{M}}:=\{(U_1,\hdots,U_k,X)\in\prod_{i=1}^kU(m,n)\times\mathcal{D}:e_j^\dagger U_iXU_i^\dagger e_j=0,\ j\in\{1,\dots,m\},i\in\{1,\dots,k\}\}$.

First, we fix the measure on $\mathcal{M}_{1,k}^m(\mathbb{C}^n)$: Let  $\phi$ be the map defined in equation \eqref{phi}. We define the measure $\mu$ on $\mathcal{M}_{1,k}^m(\mathbb{C}^{n})$ to be the pushforward measure of the $2knm$-dimensional Hausdroff measure $\mu_H$ on $\prod_{i=1}^kU(m,n)\subseteq \R^{2nmk}$, i.e. $\mu(A):=\phi_*(\mu_{H})(A)=\mu_H(\phi^{-1}(A))$ for $A\subseteq\mathcal{M}_{\text{vN}}^m(\mathbb{C}^{n})$ a measurable set.

Note that $\phi$ is the quotient projection with respect to the left action of the toral goup $T:=\prod_{i=1}^kT(m),\ T(m):=\{\text{diag}(\lambda_1,\hdots,\lambda_m):\lambda_i\in U(1)\}$ on $\prod_{i=1}^kU(m,n)$ given by $((U_1,\hdots,U_k),(T_1,\hdots,T_k))\mapsto((U_1,\hdots,U_k),(T_1U_1,\hdots,T_kU_k))$. Also note that the equations \eqref{eqnblaaaa} are invariant under the action of $T$ and hence $T\pi_1(\tilde{\mathcal{M}})=\pi_1(\tilde{\mathcal{M}})$ where $\pi_1:\prod_{i=1}^kU(m,n)\times \mathcal{D}\to \prod_{i=1}^k U(m,n)$ is the projection on the first factor. Thus, for $\mu_H(\pi_1(\tilde{\mathcal{M}}))=0$, we find
\begin{align*}
 \mu\left(\phi\circ\pi_1(\tilde{\mathcal{M}})\right)&=\mu_H\left(\phi^{-1}\left(\phi\circ\pi_1(\tilde{\mathcal{M}})\right)\right)
 \\&=\mu_H\left(T\pi_1(\tilde{\mathcal{M}})\right)\\
 &=\mu_H\left(\pi_1(\tilde{\mathcal{M}})\right)=0.
\end{align*}
Hence, it suffices to prove that $\mu_H(\pi_1(\tilde{\mathcal{M}}))=0$.

Finally, for $k(m-1)>\dim\mathcal{D}$ we find $\dim \pi_1(\tilde{\mathcal{M}})\leq\dim \prod_{i=1}^kU(m,n)+\dim\mathcal{D}-m(k-1)<\dim \prod_{i=1}^kU(m,n)$ by Lemma \ref{LEMvN}. So $\pi_1(\tilde{\mathcal{M}})$ has $\mu_H$-measure zero in $\prod_{i=1}^kU(m,n)$. The stability follows directly from Lemma \ref{stab}.
\end{proof}
\begin{remark}
Note that by the remark after Lemma \ref{LEMvN}, this proof just depends on $\mathcal{D}\subseteq H(\mathbb{C}^n)$ and hence naturally extends to semi-algebraic subsets $\mathcal{R}\subseteq H(\mathbb{C}^n)$. Furthermore, this proof shows that indeed $\pi_1(\tilde{\mathcal{M}})$ has $\mu_H$-measure zero in $\prod_{i=1}^kU(m,n)$. Thus the statement of Theorem \ref{THMvN} naturally also holds for tight frames $U\in U(n,m)$.
\end{remark}
\subsection*{Proof of Theorem \ref{THMlocal}}
For a given non-zero $X\in H(\h)$, consider the equations
\begin{align}\label{eqnblaaaaa}
p^i((O^1_1,\dots,O^1_k),\dots,(O^m_1,&\dots,O^m_k)):=\text{tr}((O^i_1\otimes\dots\otimes O^i_k)X)=0,\ i\in\{1,\hdots,m\},
\end{align}
in $((O^1_1,\dots,O^1_k),\dots,(O^m_1,\dots,O^m_k))\in (\Pi_{i=1}^kH(\mathbb{C}^{n_i}))^m$. Under the identification 
$H(\mathbb{C}^{n_i})\simeq \R^{n_i^2}$, these equations can be considered as real algebraic equations in the variables $((O^1_1,\dots,O^1_k),\dots,(O^m_1,\dots,O^m_k))$.
The following Lemma is the analogue of Lemma \ref{LEMvN}.
\begin{lemma}\label{LEMlocal}
Let $X\in H(\h)$ be non-zero. Imposing the equations \eqref{eqnblaaaaa} on $(\Pi_{i=1}^kSH(\mathbb{C}^{n_i}))^m$ decreases the dimension by at least $m$.
\end{lemma}
\begin{proof}
The equation $p_i$ just involves the variables $(O^i_1,\dots,O^i_k)$ of the  $i$-th factor of $(\Pi_{i=1}^kH(\mathbb{C}^{n_i}))^m$. Thus, it suffices to prove that, for given non-zero $X\in H(\h)$, imposing the equation
\begin{align}\label{abc}
p((O_1,\dots,O_k)):=\text{tr}((O_1\otimes\dots\otimes O_k)X)=0
\end{align}
on $\Pi_{i=1}^kSH(\mathbb{C}^{n_i})$ decreases the dimension by at least one.

In order to see that this is true, note that there are $(O_1,\dots,O_k)\in\Pi_{i=1}^kSH(\mathbb{C}^{n_i})$ such that $\text{tr}((O_1\otimes\dots\otimes O_k)X)\neq 0$ because $\bigotimes_{i=1}^kH(\mathbb{C}^{n_i})$ has a basis of normalized local operators and $X\neq 0$. But then, the equation \eqref{abc} is a non-trivial algebraic equation on the irreducible algebraic set $\Pi_{i=1}^kSH(\mathbb{C}^{n_i})$ and thus the dimension has to decrease since for a proper algebraic subset $V$ of an irreducible algebraic set $W$ we have $\dim V<\dim W$.
\end{proof}
\begin{remark}
By going along the lines of the this proof, it is easily seen that Lemma \ref{LEMlocal} also holds when going from hermitian matrices to traceless hermitian matrices, i.e. if we replace $(\Pi_{i=1}^kSH(\mathbb{C}^{n_i}))^m$ by $(\Pi_{i=1}^kSH(\mathbb{C}^{n_i})_0)^m$. Furthermore, the proof of Theorem \ref{THMlocal} also holds when going from $(\Pi_{i=1}^kSH(\mathbb{C}^{n_i}))^m$ to $(\Pi_{i=1}^kSH(\mathbb{C}^{n_i})_0)^m$ and considering $H_{loc,0}(\h):=\{O_1\otimes \hdots\otimes O_k:O_i\in SH(\mathbb{C}^{n_i})_0\}$ instead of $H_{loc}(\h)$.
\end{remark}
Now we can give the proof of Theorem \ref{THMlocal}.
\begin{proof}
Let $\psi$ be the map defined in \eqref{psi}. We can assume that $\mathcal{D}$ is a closed subset of $SH(\h)$ because if not we can replace it by the closure of $\psi(\mathcal{D})$ without increasing its dimension \footnote{See remark after Lemma \ref{stab} for more details.}. Let $\mathcal{M}$ be the semi-algebraic set obtained from $(\Pi_{i=1}^kH(\mathbb{C}^{n_i}))^m\times \mathcal{D}$ by imposing the equations \eqref{eqnblaaaaa}.

For $m>\dim\mathcal{D}$ we get $\dim \pi_1(\mathcal{M})<\dim (\Pi_{i=1}^kSH(\mathbb{C}^{n_i}))^m$ by Lemma \ref{LEMlocal}. 

Now consider $\theta(\pi_1(\mathcal{M}))$ where 
\begin{align*}
\theta:(\Pi_{i=1}^kSH(\mathbb{C}^{n_i}))^m&\to(H_{loc}(\h))^m,\\
 (O^1_1,\dots,O^1_k),\dots,(O^m_1,\dots,O^m_k)&\mapsto (O^1_1\otimes\dots\otimes O^1_k),\dots,(O^m_1\otimes\dots\otimes O^m_k).
\end{align*}
Note that $\theta$ is a surjective semi-algebraic map and thus $(H_{loc}(\h))^m$ is semi-algebraic with $\dim((H_{loc}(\h))^m)\le \dim((\Pi_{i=1}^kSH(\mathbb{C}^{n_i}))^m).$ Furthermore, $\theta$ is injective when restricting to positive matrices and hence $d:=\dim (H_{loc}(\h))^m=\dim(\Pi_{i=1}^kSH(\mathbb{C}^{n_i}))^m$. 

Finally, since $\dim \pi_1(\mathcal{M})<d$ and $\theta$ is semi-algebraic, we have $\dim\left(\theta(\pi_1(\mathcal{M}))\right)<d$ and thus $\theta(\pi_1(\mathcal{M}))$ has zero $d$-dimensional Hausdorff measure. Stability follows directly from Lemma \ref{stab}.
\end{proof}

\subsection*{Proof of Theorem \ref{tightframe}}\label{prooftightframe}
{
\begin{proof}
By going along the lines of the proof of Lemma \ref{set}, it is easily seen that $\mathcal{D}:=\{X\in\mathcal{P}_{2r}(\h):\ \text{tr}(X^2)=2\}$ represents $\Delta(\mathcal{P}_r^n)-\{0\}$ and furthermore we have $\dim \mathcal{D}=4r(n-r)-1$ by Corollary \ref{corrank} \footnote{Note that the definition of a representing set naturally generalizes to subsets $\mathcal{R}\subseteq H(\mathbb{C}^n)$.}. Applying Theorem \ref{THMvN} \footnote{Theorem \ref{THMvN} also applies in this situation. See the remark after proof of Theorem \ref{THMvN} for more details. } to the set $\mathcal{D}$ then concludes the proof.
\end{proof}
}

\appendix
\section{Hausdorff Measure on Semi-Algebraic Sets}\label{appendixD}
The term ''almost'' all used in many of the results of the present article refers to the Hausdorff measure on real affine space. In this section we define the Hausdorff measure and we prove the well-known fact that a semi-algebraic set of dimension $d$ has zero $(d+1)$-dimensional Hausdorff measure.

For a non-empty subset $A\subseteq \mathbb{R}^n$ the diameter of $S$ is defined by $\text{diam}(S):=\sup\{\|x-y\|_2:x,y\in S\}$.

Let $m\in\mathbb{R}$. For an arbitrary subset $S\subseteq\R^n$ the $m$-dimensional Hausdorff measure $\mu_H^m(S)$ is defined by (see Section 2.3 of \cite{morgan2008geometric})
\begin{align*}
\mu_H^m(S)=\lim_{\delta\to 0}\inf\{\sum_{i=1}^\infty (\text{diam}(S_i))^m:S\subseteq\cup_{i\in\mathbb{N}}(S_i),\ \text{diam}S_i<\delta\}.
\end{align*}

\begin{proposition}\label{propdim}
Let $m>n$. A semi-algebraic set $S$ of dimension $n$ has zero $m$-dimensional Hausdorff measure.
\end{proposition}
\begin{proof}
Every $n$-dimensional semi-algebraic set $S$ can be expressed as $S=\bigcup_{i=1}^k S_i$ for some $k\in \mathbb{N}$ where the $S_i$ are diffeomorphic to $(0,1)^{n_i}$, $n_i\leq n$ (see Proposition 2.9.10 of \cite{bochnak1998real}). Let us denote these diffeomorphisms by $\phi_i:(0,1)^{n_i}\to S_i$. Since $S$ is a finite union it suffices to prove that the $m$-dimensional Hausdorff measure of $S_i$ is zero for $m>n$. 

For each point $p\in S_i$, there is a neighbourhood $N_p$ of $p$ such that $\phi_i|_{N_p}$ is Lipschitz. Constructing such neighbourhoods for all $p\in S_i$, we obtain an open cover of $S_i$ by the open sets $\{N_p\}_{p\in S_i}$ and since $\R^n$ is second countable there is a countable subcover $\{S_i\cap N_{p_j}\}_{j\in\mathbb{N}}$. 

Finally, we just have to see that the Hausdorff measure of $N_{p_j}$ is zero for all $j\in\mathbb{N}$. But $\phi_i(N_{p_j})$ is the image of a set of zero $m$-dimensional Hausdorff measure under a Lipschitz map and thus $\phi_i(N_{p_j})$ has zero $m$-dimensional Hausdorff measure as well.
\end{proof}
\begin{remark}
Note that this proof in particular  shows that the $n$-dimensional Hausdorff measure of an $n$-dimensional semi-algebraic set does not vanish and hence it is a suitable measure for our purposes.
\end{remark}
The set of measurement schemes always is a semi-algebraic subset $S$ of a real affine space and the measure we choose for $S$ is the $m$-dimensional Hausdorff measure where $m$ is the dimension of $S$. If we say that almost all elements of an $m$-dimensional semi-algebraic set $S$ has a certain property we mean that it fails to hold on a subset $A\subseteq S$ that has $m$-dimensional Hausdorff measure zero. We do this by showing that the algebraic dimension of $A$ is smaller than $m$ and applying Proposition \ref{propdim}.




\bibliographystyle{unsrt}
\bibliography{bibliography}

\end{document}